\documentclass[11pt]{article}
\usepackage[utf8]{inputenc}

\usepackage{mathtools}
\usepackage{amsthm,amssymb, mathrsfs,tikz,cleveref,comment,subfiles}
\usepackage{multirow}
\usepackage{enumerate}
\usepackage{tikz}
\usetikzlibrary{positioning}
\usepackage{listings}
\usepackage{url}
\usepackage{comment}
\usepackage{pdfpages}
\usepackage{fullpage}

\newcommand{\F}{\mathbb{F}}

\newcommand{\calC}{\mathcal{C}}
\newcommand{\calV}{\mathcal{V}}

\newcommand{\ov}[1]{\overline{#1}}

\newcommand{\Vcond}{\calV\cap\langle e_i,e_j\rangle=\{0\}\ \forall i,j\in[\mu]}

\newtheorem{theorem}{Theorem}[section]
\newtheorem{lemma}[theorem]{Lemma}

\newtheorem{corollary}[theorem]{Corollary}

\newtheorem{definition}[theorem]{Definition}

\crefname{lemma}{Lemma}{Lemmas}

\title{Batch Codes from Affine Cartesian Codes and Quotient Spaces\thanks{\textbf{The research was partially supported by the National Science Foundation under grant DMS-1547399.}}}
\author{Baumbaugh, Travis Alan \and Colgate Haley \and Jackman, Timothy \and Manganiello, Felice}
\date{May 2020}

\begin{document}

\maketitle

\begin{abstract}
    Affine Cartesian codes are defined by evaluating multivariate polynomials at a cartesian product of finite subsets of a finite field. In this work we examine properties of these codes as batch codes. We consider the recovery sets to be defined by points aligned on a specific direction and the buckets to be derived from cosets of a subspace of the ambient space of the evaluation points. We are able to prove that under these conditions, an affine Cartesian code is able to satisfy a query of size up to one more than the dimension of the space of the ambient space. \end{abstract}

\section{Introduction}

Batch codes may be used in information retrieval when multiple users want to access potentially overlapping requests from a set of devices while achieving a balance between minimizing the load on each device and minimizing the number of devices used. We can view the buckets as servers and the symbols used from each bucket as the load on each server. In the original scenario, a single user is trying to reconstruct $t$ bits of information. This definition naturally generalizes to the concept of multiset batch codes which have nearly the same definition, but where the indices chosen for reconstruction are not necessarily distinct.

The family of codes known as batch codes was introduced in \cite{Ishai04}. They were originally studied as a scheme for distributing data across multiple devices and minimizing the load on each device and total amount of storage consumed. In this paper, we study $[n,k,t,m,\tau]$ batch codes, where $n$ is the code length, $k$ is the dimension of the code, $t$ is the number of entries we wish to retrieve, $m$ is the number of buckets, and $\tau$ is the maximum number of symbols used from each bucket for any reconstruction of $t$ entries. We seek to minimize the load on each device while maximizing the amount of reconstructed data. That is, we want to minimize $\tau$ while maximizing $t$.

This corresponds to $t$ users who each wish to reconstruct a single element, among which there may be duplicates. This is similar to Private information retrieval (PIR) codes, which differ in that $t$ duplicates of the same element must be reconstructed. Other schemes dealing with multiple requests are addressed in \cite{Ramakrishnan18}. For batch and PIR codes where the queries do not all necessarily occur at the same time, see \cite{Riet18}. Restricted recovery set sizes are considered in \cite{Thomas17}. Another notable type of batch code defined in \cite{Ishai04} is a primitive multiset batch code where the number of buckets is $m=n$.

Much of the related research involves primitive multiset batch codes with a systematic generator matrix. In \cite{Ishai04}, the authors give results for some multiset batch codes using subcube codes and Reed-Muller codes. They use a systematic generator matrix, which often allows for better parameters. Their goal was to maximize the efficiency of the code for a fixed number of queries $t$. The focus of research on batch codes then shifted to combinatorial batch codes. These were first introduced by \cite{paterson09}. They are replication-based codes using various combinatorial objects that allow for efficient decoding procedures. We do not consider combinatorial batch codes, but some relevant results can be found in \cite{paterson09}, \cite{Bujtas11}, \cite{Bhattacharya12}, and \cite{Silberstein16}.

In order to reduce wait time for multiple users, we may look at locally repairable codes with availability as noted in \cite{Dimakis11}. A locally repairable code, with locality $r$ and availability $\delta$, provides us the opportunity to reconstruct a particular bit of data using $\delta$ disjoint sets of size at most $r$ \cite{Skachek18}. When we only need to reconstruct this one bit multiple times, this gives us properties of the code as a private Information Retrieval (PIR) code. However, the research in this paper covers the scenario in which some bits may differ.

The Hamming weights of affine Cartesian codes are studied in \cite{Beelen18}. This is a generalization of work in \cite{Heijnen98}, and in a similar fashion, the work in this paper aims to expand the study of batch properties from Reed-Muller codes as studied in \cite{prevSummer} to the broader class of affine Cartesian codes. In the same manner, we begin by examining codes with $\tau=1$. The even broader family of generalized affine Cartesian codes, specifically those with complementary duals, are studied in \cite{Lopez19}.

This work focuses on studying the properties of affine Cartesian codes as batch codes. In Section \ref{batch_background}, we formally introduce batch codes and affine Cartesian codes. 
In Section \ref{cartRecovery},  we define the special recovery sets for affine Cartesian codes based on the points in the direction of a coordinate. The main body is Section \ref{quotientBucket}; there we define the building blocks of a batch code - the buckets. In this work we suggest the buckets to be cosets of a subset $V$ of $\F_q^\mu$. Under several equivalent conditions, we show in Theorem \ref{Expanding} that a maximal length affine Cartesian code can satisfy queries of size up to $t=n+1$. The specific case with $V=\langle(1,1,\dots,1)\rangle$ is considered in Subsection \ref{diagSubspace}. We conclude by generalizing the result to any affine Cartesian code.

\section{Background}\label{batch_background}

Batch codes were introduced in \cite{Ishai04}.
Throughout this work, by batch codes we refer specifically to multiset batch codes, defined by \cite{Ishai04}. To build up to this definition, we first introduce several notions. For a given linear code $\calC\subseteq \F_q^n$, where $\F_q$ represents the finite field of $q$ elements, and the index set $[n]:=\{1,\dots, n\}$ we give the following definitions.

\begin{definition} A {\em bucket configuration} $B_1,\dots,B_m$ is a partition on index set $[n]$. For each $k\in[m]$, the $B_k$ is referred to as a {\em bucket}.
\end{definition}

\begin{definition}\label{recoverySet}
	For any index $i\in [n]$, a \em{recovery set} $R_i$ for the index $i$ is a set such that, for any codeword $c\in\calC$, the value of $c_i$ may be recovered by reading the symbols $\{c_j\mid j\in R_i\}$.
\end{definition}
At this point, we note two distinct categories of recovery sets. If $R_i=\{i\}$, then we refer to $R_i$ as {\em direct access}. If instead $i\notin R_i$, then we refer to $R_i$ as an {\em indirect} recovery set. We also note that while any set containing a recovery set is technically a recovery set, these shall not be considered proper recovery sets in the remainder of the paper. Now we deal with multiple recovery sets at the same time for a query of indices that are not necessarily distinct. 

\begin{definition}\label{SolutionSet}
	Given a query $Q=(i_1,\dots,i_t)\in [n]^t$, we say that a set of recovery sets $R^Q=\{R_{i_1},\dots,R_{i_t}\}$ is a {\em query recovery set} with property $\tau$ for $Q$ if
	\begin{enumerate}
		\item $\left|\left(\bigcup_{s=1}^{t}R_{i_s}\right)\cap B_k\right|\le\tau$ $\forall$ $k \in [m]$, and
		\item $R_{i_r}\cap R_{i_s}=\emptyset$ $\forall$ $r,s\in[t]$ where $r\ne s$.
	\end{enumerate}
\end{definition}

\begin{definition}\label{ValidBatchGenerated}
	We say that a bucket configuration $B_1,\dots, B_m$ is {\em $t,\tau$ valid} if, for all queries $Q=(i_1, \dots,i_t)\in [n]^t$, there exists a query recovery set $R^Q$ with property $\tau$.
\end{definition}

Now, with the building blocks in place, we may more rigorously define batch codes.

\begin{definition}
	A $[n,k, t, m, \tau]$ linear {\em batch code} $\calC$ over $\F_q$ is a linear code $\calC$ of length $n$ and dimension $k$, together with a $t,\tau$ valid bucket configuration $B_1,\dots,B_m$.
\end{definition}

Throughout this work, we will focus on the case  $\tau=1$. The following lemma, proven in \cite{Ishai04}, allows us to do this.

\begin{lemma}\label{lemucla2}
	Any $[n,k,t,m,1]$ batch code is also an $[n,k,t,\lceil\frac{m}{\tau}\rceil,\tau]$ batch code.
\end{lemma}

An affine Cartesian code is defined as follows.

\begin{definition}
	Let $\F_q$ be an arbitrary field, and $A_1, \ldots,
        A_{\mu}\subseteq \F_1$ be non-empty subsets. Define $X$ to be
        the cartesian product $A_1 \times \ldots \times A_{\mu}\subseteq\F_q^{\mu}$. Let $S=\F_q[x_1, \ldots, x_{\mu}]$ be a multivariate polynomial ring and $S^{\le\rho}$ be the subspace of $S$ of all polynomials  with total degree at most $\rho$. Let $X=\{p_1, \ldots ,p_n\}$. The affine Cartesian code $\calC_{X}(\rho)$ of degree $\rho$ is  
    \[\calC_{X}(\rho)=\left\{(f(p_1),\dots,f(p_n))\mid f\in S^{\le \rho} \right\},\]	
	meaning the image of the evaluation map: 
	\begin{align*}
	ev_{\rho}: S^{\le\rho}&\rightarrow\F_q^n\\
	f&\mapsto (f(p_1), \ldots, f(p_n)).
	\end{align*}
	
\end{definition}

The affine Cartesian code $\calC_{X}(\rho)$ is an $[n,\rho+1]$-linear code. The distance of such codes is studied in  \cite{Lopez14}. Note that affine Cartesian codes are a generalization of Reed-Muller codes since $\calC_{\F_q^\mu}(\rho)$ is a Reed-Muller code. For these codes, the index set of the code corresponds $X$. This is important in the construction of batch codes based on affine Cartesian codes as the entries of the codewords correspond to the evaluations polynomials in the points of $X$. 

In the following section, we investigate the recovery sets for an affine Cartesian code that arise from using the structure of the set $X$.

\section{Recovery Sets}\label{cartRecovery}

This section focuses on the characterization of some recovery sets for $\calC_{X}(\rho)$.
It is well known that a univariate polynomial of degree $\rho-1$ is uniquely determined by $\rho$ of its evaluations. Furthermore, one can find the polynomial starting from the evaluations by using Lagrange interpolation. The following holds.





\begin{lemma}\label{recoverable}
Let $X=A_1\times\dots \times A_\mu$. For  $p=(a_1,\dots,a_{\mu})\in X$, and $i\in[\mu]$, let \[R_{p,i}=\left\{(b_1,\dots,b_\mu)\mid b_i\in A_i, \ b_j=a_j \ \mbox{if} \ j\neq i\right\} \setminus\{p\}.\]
If $\rho+1<|A_i|$, then for any $f\in S^{\le\rho}$, the value of $f(p)$ can be recovered using the values $f(R_{p,i})$.
\end{lemma}
\begin{proof}
	Let $f\in S^{\le\rho}$, where $\rho+1\le|A_i|$. By evaluating the multivariate polynomial $f$ in all but the $i$th of the coordinates of $p$, we obtain the univariate polynomial     $f_i(x_i)=f(a_1,\dots,a_{i-1},x_i,a_{i+1},\dots,a_{\mu})\in\F_q[x_i]$. By the construction of $R_{p,i}$, we have that $f(R_{p,i})=f_i(A_i\setminus\{a_i\})$. Since $f$ is a polynomial of total degree at most $\rho$, $f_i(x_i)$ is a polynomial of degree at most $\rho$ in $x_i$. If $\rho+1<|A_i|$, then $\rho\le|A_i|-2$, so $f_i$ is of degree at most $|A_i|-2$. By Lagrange interpolation, we may find a unique polynomial $g(x_i)\in\F_q[x_i]$ of degree at most $|A_i\setminus\{a_i\}|-1=|A_i|-2$ such that $g(a)=f_i(a)$ for all $a\in A_i\setminus\{a_i\}$, and so we must have $g=f_i$. We find that $g(a_i)=f_i(a_i)=f(a_1,\dots,a_{\mu})=f(p)$, and so we can recover $f(p)$.
\end{proof}

We already recalled that affine Cartesian codes are a generalization of Reed-Muller codes. For simplicity of notation, the remainder of this section and the next section focus on Reed-Muller codes. It is only at the end of the paper that we generalize the results to affine Cartesian codes. We thus initially consider  $\calC_{\F_q^\mu}(\rho)$ and $\rho<q-1$. 

From to the previous lemma, we obtain the following characterization of some recovery sets for $\calC_{\F_q^\mu}(\rho)$.

\begin{corollary}\label{qRecovery}
    For a $p\in \F_q^\mu$, the sets $R_{p,i}=(p+\langle e_i\rangle) \setminus \{p\}$ for $i\in[\mu]$ are recovery sets for $p$ in $\calC_{\F_q^\mu}(\rho)$.
\end{corollary}


As in Section \ref{batch_background}, we refer to $R_{p, i}$ for $i \in [\mu]$ as an indirect recovery of $p$. These sets are in correspondence with one-dimensional affine spaces of $\F_q^{\mu}$, where $i$ corresponds to the only index that has varying entries. The direct access of $p$ is represented by $R_{p, 0}=\{p\}$. 

\begin{corollary}
	For any query $Q=(p_1,\dots,p_{\mu+1})\in(\F_q^{\mu})^{\mu+1}$, using the indices in a query recovery set $R^Q=\{R_{p_1,i_1},\dots,R_{p_{\mu+1},i_{\mu+1}}\}$, it is possible to recover $f(p_1),\dots,f(p_{\mu+1})$, where $f\in S^{\le \rho}$. 
\end{corollary}
\begin{proof}
	For any $s\in[\mu+1]$ such that $i_s=0$, we note that $R_{p_s,i_s}=R_{p_s,0}=\{p_s\}$, and so this is direct access, and we may simply calculate $f(p_s)$. That these are recovery sets for $p_s$ such that $i_s\ne0$ follows from Lemma \ref{recoverable}, noting that with $X=\F_q^\mu$, the two definitions of $R_{p_j,i_j}$ coincide.
\end{proof}

Note that for the previous corollary $\tau$ is not necessary equal to 1.

For the rest of the work we will consider query recovery sets consisting only of $R_{p,i}$ recovery sets for $i=0,\dots,\mu$.  We will leave off the $Q$ in $R^Q$ when the context makes the query unambiguous. To be more precise about batch properties, we restate the conditions that every query recovery set must satisfy for a bucket configuration to be valid with $t=\mu+1$ and $\tau=1$, the parameters we will be using in the following section:
\begin{align}
    & \left|\left(\bigcup_{s=1}^{\mu+1}R_{p_s,i_s}\right)\cap B_k\right|\le1 &&\forall k\le m, \label{intersection}\\
	& R_{p_r,i_r}\cap R_{p_s,i_s}=\emptyset& &\forall r,s\in[\mu+1], \ r\ne s. \label{empty}
\end{align}

The first condition corresponds to using at most $\tau=1$ indices in any given bucket, while the second corresponds to having non-overlapping recovery sets. 

\section{Quotient-Space Bucket configuration}\label{quotientBucket}

With requirements for valid bucket configurations addressed, we now define the bucket configuration used in this paper.
\begin{definition}\label{quotient-buckets}
	For any subspace $\calV$ of $\F_q^{\mu}$, consider the quotient space $\F_q^{\mu}/\calV$. The equivalence classes $[p]=p+\calV$ partition $\F_q^{\mu}$. We define a quotient-space bucket configuration to be one where the buckets are these equivalence classes.
\end{definition}

Note that although the notation $[\cdot]$ is used both for equivalence classes and index sets, its use will be clear from context. 

\begin{definition}
	We denote by $\sim$ the equivalence relation on $\F_q^\mu$ induced by $\calV$, that is, $p_1\sim p_2$ if and only if $p_1-p_2\in \calV$.
\end{definition}

With this bucket configuration, we have $m=q^{\mu-\dim{\calV}}$. This configuration provides us with a great deal of symmetry and structure, which allows us to approach determining the validity of a given quotient-space bucket configuration with the following tools.

For any $p\in\F_q^{\mu}$, the set $[p]=p+\calV$ is all elements in the same bucket as $p$ by definition. 
	For any subset $U\subseteq\F_q^n$, 
	$[U]=\{[p]\mid p\in U\}$ is the set of all buckets corresponding to points in $U$.  With this notation, we now note an important result with respect to recovery sets for equivalent points.

\begin{lemma}\label{CollapsibleBucket}
	With the subspace construction, if $p_1\sim p_2$, then  $[R_{p_1,i}]=[R_{p_2,i}]$ for all $i\in[\mu]$.
\end{lemma}

\begin{proof}
	For all $i\in[\mu]$ and any $\alpha\in\F_q$, $(p_1+\alpha e_i)-(p_2+\alpha e_i)=p_1-p_2\in \calV$, so $[p_1+\alpha e_i]=[p_2+\alpha e_i]$. Thus we may write \[[R_{p_1,i}]=\{[p_1+\alpha e_i]\mid\alpha\in\F_q\setminus\{0\}\}=\{[p_2+\alpha e_i]\mid\alpha\in\F_q\setminus\{0\}\}=[R_{p_2,i}].\]
\end{proof}

This means that under the equivalence relation, the recovery sets for elements in the same bucket are the same. This identical use of buckets for the recovery sets leads to the following:

\begin{corollary}\label{queryEquiv}
	Let $Q=(p_1,\dots,p_{\mu+1})\in(\F_q^{\mu})^{\mu+1}$ be a query such that $p_{i_1}\sim p_{i_2}$ for some $i_1\neq i_2$. Let $Q'=(p_1',\dots,p_{\mu+1}')$ be a query where $p_{i_2}'=p_{i_1}$ and $p_i'=p_i$ for $i\neq i_2$. Then $R$ is a query recovery set for $Q'$ if and only if it is a query recovery set for $Q$.
\end{corollary}

In other words, we may effectively treat recovering multiple equivalent points in the same bucket as recovering the same point multiple times. This leads naturally to notation for all recovery sets of a point.

\begin{definition}
	For  $p\in\F_q^n$, define $R_p=\{R_{p,0},\dots,R_{p,\mu}\}$ and $E_p=\bigcup_{i=0}^{\mu}R_{p,i}$.
\end{definition}

We now reach the central theorem which will be used to verify the validity of quotient-space bucket configurations.

\begin{theorem}\label{equiv}
	The following are equivalent:
	\begin{enumerate}[i)]
		\item\label{vcond} $\Vcond$.
		\item\label{forall} For $p\in\F_q^{\mu}$, if $a,b\in E_p$ are distinct, then $[a]\ne[b]$.
		\item\label{recovery} For $p\in\F_q^{\mu}$, $R_p$ is a query recovery set for $Q=(p,\dots,p)$.
	\end{enumerate}
\end{theorem}

\begin{proof}
We proceed with proving the equivalences.
\begin{itemize}

    \item[$\ref{vcond})\Rightarrow \ref{forall}$)] Let $a,b\in E_p$ be distinct elements for some $p\in \F_q^\mu$. Then, $a=p+\alpha e_i$ and $b=p+\beta e_j$ for some $i,j\in [\mu]$. Since $a\neq b$, it holds that $a-b=\alpha e_i-\beta e_j \neq 0$ which implies that $a-b\not\in \calV$ as per \textit{\ref{vcond})}, implying that $[a]\neq [b]$. 
    
    \item[$\ref{forall})\Rightarrow\ref{vcond})$] 
    We prove this implication by contraposition. Suppose that there exist $\alpha,\beta \in \F_q$ not both zero such that $\alpha e_i - \beta e_j\in \calV$ with $i\neq j$. For arbitrary $p\in \F_q^{\mu}$, define $a=p+\alpha e_i$ and $b=p+\beta e_j$. Then $a\neq b$, but since $a-b=\alpha e_i - \beta e_j\in \calV$, we have $[a]=[b]$. 
    
    \item[$\eqref{forall}\Leftrightarrow\eqref{recovery}$] Since $E_p$ is the union of the sets in $R_p=\{R_{p,i}\mid 0\le i\le \mu\}$, suppose $R_p$ is a query recovery set. Given that $\tau=1$, by Condition \eqref{intersection} each point in the union of these sets must be in a separate bucket. So for all $a,b\in E_p$, $a\ne b\implies [a]\ne[b]$. Similarly, if each point in $E_p$ is in a different bucket, then Condition \eqref{intersection} is satisfied, and the sets $R_{p,0},\dots,R_{p,\mu}$ are all disjoint by construction, so Condition \eqref{empty} is satisfied. This makes $R_p$ a query recovery set for $Q=(p,\dots,p)$.
\end{itemize}


\end{proof}

Note that the first condition of this lemma implies that we need $\mu\ge 3$.
These equivalent conditions lead to some important necessary conditions.
\begin{corollary}
	An affine Cartesian code $\calC_{\F_q^\mu}(\rho)$ is a valid batch code with quotient-space bucket configuration induced by $\calV\subset\F_q^{\mu}$ only if $\Vcond$.
\end{corollary}


Next, we specify a way that a batch code over $\F_q^{\mu-1}$ may be expanded to a batch code over $\F_q^{\mu}$. This is then used in constructing a bucket configuration for affine Cartesian codes by induction.

\begin{theorem}\label{Expanding}
	Consider the puncturing function $\phi:\F_q^\mu\to\F_q^{\mu-1}$ defined by $\phi((a_1,\dots,a_\mu))=(a_1,\dots,a_{\mu-1})$ and let $\calC=\calC_{\F_q^{\mu}}(\rho)$ and $\ov{\calC}=\calC_{\F_q^{\mu-1}}(\rho)$. 
	Let $\calV\subseteq\F_q^{\mu}$ and $\ov{\calV}\subseteq\F_q^{\mu-1}$ be subspaces such that $\phi(\calV)=\ov{\calV}$ and $\Vcond$. 
	If $\ov{\calC}$ is a batch code with quotient-space bucket configuration induced by  $\ov{\calV}$ and  $t=\mu$, then $\calC$ is a batch code with quotient-space bucket configuration induced by  $\calV$ and $t=\mu+1$.
\end{theorem}

\begin{proof}
	We claim that for  $a\in\F_q^{\mu}$, it holds that $\phi([a])\subseteq[\phi(a)]$. If  $b\in[a]$, it holds that $a-b\in \calV$. This means that $\phi(a-b)\in\phi(\calV)=\ov{\calV}$, and by linearity of $\phi$, we have $\phi(a)-\phi(b)\in\ov{\calV}$. This in turn means $\phi(b)\in[\phi(a)]$. Since this is true for  $\phi(b)\in\phi([a])$, we have that $\phi([a])\subseteq[\phi(a)]$. From this, we see that $a\sim b\Rightarrow [b]=[a]\Rightarrow\phi([b])\subset[\phi(a)]$. In particular, $\phi(b)\in\phi([b])$, and since $\phi(b)\in[\phi(a)]\Rightarrow[\phi(a)]=[\phi(b)]$, we have that $a\sim b\Rightarrow\phi(a)\sim\phi(b)$.
	
	Now consider any query $Q=(p_1,\dots,p_\mu,p_{\mu+1})$ of points in $\F_q^{\mu}$. The multiset $Q'=(\phi(p_1),\dots,\phi(p_{\mu}))$ is a query of $\mu$ elements in $\F_q^{\mu-1}$. For any $a\in\F_q^{\mu}$, let $\ov{a}=\phi(a)\in\F_q^{\mu-1}$. Then we may write $Q'=(\ov{p_1},\dots,\ov{p_{\mu}})$, and since $\ov{\calC}$ is a batch code that can satisfy any query of size $t=\mu$, there exists some query recovery set $\ov{R}=\{R_{\ov{p_1},i_1},\dots,R_{\ov{p_{\mu}},i_{\mu}}\}$ such that
	\begin{enumerate}
		\item $\left|\left(\bigcup_{s= 1}^n R_{\ov{p_s},i_{\ell}}\right)\cap[b]\right|\leq1$ $\forall\,[b]\in\F_q^{n-1}/\ov{\calV}$
		\item $R_{\ov{p_r},i_r}\cap R_{\ov{p_s},i_s}=\emptyset$ $\forall$ $r,s\in[\mu]$ where $r\ne s$.
	\end{enumerate}
	
	Now let $E=\cup_{s=1}^\mu R_{p_s,i_s}$. If there exists some $z\in E$ such that $z\sim p_{\mu+1}$, then let $i_{\mu+1}=\mu$, otherwise let $i_{\mu+1}=0$. We claim that
	\[R=\{R_{p_1,i_1},\dots,R_{p_{\mu},i_{\mu}},R_{p_{\mu+1},i_{\mu+1}}\}\]
	is a valid recovery set for the query $Q=(p_1,\dots,p_{\mu+1})$. 
	
	First, we check Condition \eqref{intersection} for $R$. Assume by contradiction that for some some $[c]\in\F_q^{\mu}/\calV$,  there exist distinct $a,b\in\bigcup_{\ell=1}^{\mu+1}R_{p_s,i_s}\cap[c]$. Since $a,b\in [c]$, we have $a\sim b$. As seen before, this means that $\phi(a)\sim\phi(b)$. We consider two cases: $\phi(a)=\phi(b)$ or $\phi(a)\neq \phi(b)$. 
	
	If $\phi(a)=\phi(b)$, then by definition of $\phi$, we see that $a-b\in\langle e_n\rangle$. This means that $b\in E_a$. Since $\Vcond$, part \textit{\ref{forall})} of Theorem \ref{equiv} implies that $[a]\ne[b]$, a contradiction. 
	
	Thus, instead, we must have $\phi(a)\ne\phi(b)$, which we write as $\ov{a}\ne\ov{b}$. This leads to a few possibilities. Ab before, we differentiate two cases: either $a,b\in E=\bigcup_{s=1}^{\mu}R_{p_s,i_s}$ or without loss of generality $a\in E$ and $b\in R_{p_{\mu+1},i_{\mu+1}}$.

	If $a,b\in E=\bigcup_{s=1}^{\mu}R_{p_s,i_s}$, then $\ov{a},\ov{b}\in\bigcup_{s=1}^{\mu}R_{\ov{p_s},i_s}$. This leads to a contradiction to Condition \eqref{intersection} for $\ov{R}$, as $\ov{a}$ and $\ov{b}$ are in the same bucket. By construction, $\phi(R_{p_{\mu+1},i_{\mu+1}})=\{\phi(p_{\mu+1})\}$, so $a,b\in R_{p_{\mu+1},i_{\mu+1}}$ implies $\ov{a},\ov{b}\in\phi(R_{p_{\mu+1},i_{\mu+1}})=\{\phi(p_{\mu+1})\}$, or $\ov{a}=\ov{b}=\phi(p_{\mu+1})$, and we would have a contradiction to $\ov{a}\ne\ov{b}$.
	
	Thus, we suppose without loss of generality that $a\in E$, and $b\in R_{p_{\mu+1},i_{\mu+1}}$. There are two possibilities. If $i_{\mu+1}=0$, then by our selection of $i_{\mu+1}$, there is no $z\in E$ such that $p_{\mu+1}\sim z$, but $b\in R_{p_{\mu+1},0}=\{p_{\mu+1}\}$, so $b=p_{\mu+1}$, which means $a\sim b=p_{\mu+1}$, a contradiction. If instead $i_{\mu+1}=\mu$, this means $\exists z\in E$ such that $z\sim p_{\mu+1}$. Then there is some $r\in[\mu]$ such that $z\in R_{p_r,i_r}$, and since $a\in E$, we also have some $s\in[\mu]$ such that $a\in R_{p_s,i_s}$. 
	This again leads to two possibilities: either $s=r$ or $s\neq r$. 
	
	Suppose $s=r$. Since $b\in R_{p_{\mu+1},i_{\mu+1}}$, we have that $b\ne p_{\mu+1}$ and also $b\in E_{p_{\mu+1}}$. We also have $p_{\mu+1}\in E_{p_{\mu+1}}$, so $b\ne p_{\mu+1}\implies b\not\sim p_{\mu+1}$ by Theorem \ref{equiv}. Since $z\sim p_{\mu+1}$, we must have $b\not\sim z$, or transitivity of $\sim$ would break down. Since $a\sim b$, this also means that $a\not\sim z$, so certainly $z\ne a$. Since $s=r$, we have $a,z\in R_{p_r,i_r}=R_{p_s,i_s}$, so $a-z=\alpha e_{i_r}$ for some $\alpha\in\F_q\setminus\{0\}$. We also have $b=p_{\mu+1}+\beta e_{\mu}$ for some $\beta\in\F_q\setminus\{0\}$, and we consider that $a\sim b$ and $z\sim p_{\mu+1}$. We may combine these as $a-z\sim b-p_{\mu+1}$, or $\alpha e_{i_r}\sim \beta e_{\mu}$. But this means that $\alpha e_{i_r}-\beta e_{\mu}\in \calV$. The only way this would not be a contradiction to $\Vcond$ is if $\alpha e_{i_r}=\beta e_{\mu}$, but that would mean $i_r=\mu$, which is impossible given our construction.
	
	Thus, we consider $s\ne r$. If $\ov{a}=\ov{z}$, then $\ov{a}\in R_{\ov{p_s},i_s}$, and $\ov{a}=\ov{z}\in R_{\ov{p_r},i_r}$, so $\ov{a}\in R_{\ov{p_r},i_r}\cap R_{\ov{p_s},i_s}$. This is a contradiction to \eqref{empty} for $\ov{R}$. If instead $\ov{a}\ne\ov{z}$, note that $\ov{b}=\ov{p_{\mu+1}}\sim\ov{z}$. This means that $\ov{a}\sim\ov{b}\sim\ov{z}$, so $\ov{a}\sim\ov{z}$, and this is a contradiction to Condition \eqref{intersection} for $\ov{R}$. 
	
	This concludes the proof that Condition \eqref{intersection} holds for $R$. Next, we show that Condition \eqref{empty} holds for $R$. 
	
	Again, consider the possibilities for a contradiction. If $R_{p_r,i_r}\cap R_{p_s,i_s}\ne\emptyset$ for some $r,s\in[\mu+1]$ such that $r\ne s$, then there are two possibilities: either $r,s\in[\mu]$ or without loss of generality $r\in[\mu]$ and $s=\mu+1$. 
	
	If $r,s\in[\mu]$, then this would mean that there exists some $ p\in\F_q^{\mu}$ such that $p\in R_{p_r,i_r}\cap R_{p_s,i_s}$. But then $\ov{p}\in R_{\ov{p_r},j_r}$ and $\ov{p}\in R_{\ov{p_s},i_s}$, a contradiction to Condition (2) for $\ov{R}$. 
	
	Thus, without loss of generality, we have $r\in[\mu]$ and $s=\mu+1$. There are again two possibilities: either $i_{\mu+1}=0$ or $i_{\mu+1}=\mu$.
	
	If $i_{\mu+1}=0$, this means $R_{p_s,i_s}=R_{p_{\mu+1},0}=\{p_{\mu+1}\}$, and so the intersection must be $\{p_{\mu+1}\}$. This would mean $p_{\mu+1}\in R_{p_r,i_r}$ and so $p_{\mu+1}\in E$. Since $p_{\mu+1}\sim p_{\mu+1}$ by the reflexive property of $\sim$, it is a $z\in E$ such that $z\sim p_{\mu+1}$, a contradiction to $i_{\mu+1}=0$.
	
	If instead $i_{\mu+1}=\mu$, we have some $z\in E$ such that $z\sim p_{\mu+1}$, so $z\in R_{p_k,i_k}$ for some $k\in[\mu]$. Since $R_{p_r,i_r}\cap R_{p_{\mu+1},i_{\mu+1}}\ne\emptyset$, we also have some $a\in R_{p_{\mu+1},\mu}$ such that $a\in R_{p_r,i_r}$. As before, if $k=r$, we have $a-z=\alpha e_k$ for some $\alpha\in\F_q\setminus\{0\}$, and we have $a-p_{\mu+1}=\beta e_{\mu}$ for some $\beta\in\F_q\setminus\{0\}$. But then $a-z\sim a-p_{\mu+1}$, so $\alpha e_k\sim\beta e_{\mu}$, which leads to a contradiction as before.
	
	This leaves $k\ne r$. Again, if $\ov{a}=\ov{z}$, this leads to a contradiction to (2) for $\ov{R}$, and if $\ov{a}\ne\ov{z}$, we still have $a\in R_{p_{\mu+1},\mu}$, so $\ov{a}\in\phi(R_{p_{\mu+1},\mu})=\{\ov{p_{\mu+1}}\}$, which means $\ov{a}=\ov{p_{\mu+1}}\sim\ov{z}$, and so this contradicts Condition \eqref{intersection} for $\ov{R}$. 
	
	We have once again contradicted all alternatives, so Condition \eqref{empty} must be satisfied for $R$. Because $R$ satisfies both conditions, $R$ is a valid query recovery set. Since this holds for any query $Q=(p_1,\dots,p_{\mu+1})$ of $\mu+1$ elements in $\F_q^{\mu}$, $\calC$ is a batch code that can satisfy $t=\mu+1$ requests.
\end{proof}

\subsection{The Case of a Diagonal Subspace}\label{diagSubspace}
Using the subspace $\calV=\langle(1,\dots,1)\rangle\subset\F_q^{\mu}$, we are able to generate a valid bucket configuration as long as $q\ge 3$ and $\mu\ge 3$:

\begin{theorem}\label{diagBatch}
	Let $q\ge 3$, $\mu\ge 3$ and $\calV=\langle(1,\dots,1)\rangle\subset\F_q^{\mu}$. Then $\calC_{\F_q^\mu}(\rho)$ is a batch code with quotient-space bucket configuration induced by $\calV$ and with properties $m=q^{\mu-1}$, $\tau=1$, and $t=\mu+1$.
\end{theorem}

\begin{proof}
	With this configuration, since $\dim(\calV)=1$, we have $m=q^{\mu-1}$.
	
	We begin with the base case $\mu=3$, where $\calV=\langle(1,1,1)\rangle$, and $t=4$. Since $\Vcond$, recovering four copies of any one point is possible, and by Corollary \ref{queryEquiv}, so is recovering any four points in the same bucket. Recovering four query points in different buckets is trivial using all direct access. This leaves the cases where the queries belong to either $2$ or $3$ distinct buckets. Without loss of generality and by Corollary \ref{queryEquiv}, we need to address the queries $Q=(a,a,a,b)$, $Q=(a,a,b,c)$, or $Q=(a,a,b,b)$, where $a,b,c\in \F_q^{3}$ such that $[a],[b]$, and $[c]$ are all distinct. We handle each of these separately:
	\begin{enumerate}
		\item Consider $Q=(a,a,a,b)$. If $[b]\notin[E_a]$, then any three recovery sets may be used for $a$, and $b$ may be directly accessed. Otherwise, there is one value $i\in[\mu]$ such that $[b]\in[R_{a,i}]$, and we can satisfy the request using $R=\{R_{a,0},R_{a,j_1},R_{a,j_2},R_{b,0}\}$ such that $j_1,j_2\in\left[\mu\right]\setminus i$.
		\item Recovering $Q=(a,a,b,c)$ is similar to the first case, using $R_{b,0}$ and $R_{c,0}$. Since these eliminate at most $2$ recovery sets of $a$ through intersection with $[E_a]$, there will be at least one remaining recovery set of $a$ besides the direct access.
		\item Consider $Q=(a,a,b,b)$.  Utilize $R_{a,0}$ and $R_{b,0}$. 
		If $[b]\in[R_{a,i}]$ for some $i\in\left[\mu\right]$, let $j\in\left[\mu\right]\setminus i$,  otherwise let $j$ be any $j\in\left[\mu\right]$. This means that $[b]\notin[R_{a,j}]$ by construction. To see that we can use both $R_{a,j}$ and $R_{b,j}$, assume by way of contradiction that there exists some $[p]\in[R_{a,j}]\cap[R_{b,j}]$. Then
		\[[p]=[a+\alpha e_j]=[b+\beta e_j],\]
		where $\alpha,\beta\in\F_q\setminus\{0\}$. This means $(a+\alpha e_j)-(b+\beta e_j)=a+(\alpha-\beta)e_j-b\in \calV$, which in turn means $[a+(\alpha-\beta)e_j]=[b]$. Either $\alpha=\beta$, so $[a]=[b]$, a contradiction, or $[b]\in[R_{a,j}]$, also a contradiction. Therefore $[R_{a,j}]\cap [R_{b,j}]=\emptyset$. Thus we may use $R=\{R_{a,0},R_{a,j},R_{b,0},R_{b,j}\}$.
	\end{enumerate}

	Now, assume that for some $\mu>3$, $\ov{\calV}=\langle(1,\dots,1)\rangle\subset\F_q^{\mu-1}$ generates a valid batch code with $t=\mu$. Then by Theorem \ref{Expanding}, since $\calV=\langle(1,\dots,1)\rangle\in\F_q^{\mu}$ satisfies $\Vcond$, we can extend the batch code with quotient-space buckets configuration induced by $\ov{\calV}$ into a batch code with quotient-space bucket configuration induced by $\calV$ that can satisfy any query of $t=\mu+1$ elements. By induction, we then have that for any $\mu\ge 3$, $\calV=\langle(1,\dots,1)\rangle\subset\F_q^{\mu}$ generates buckets for a batch code with $t=\mu+1$.
\end{proof}

\section{Affine Cartesian Codes}\label{generalCart}

Finally, we want to apply the techniques we have developed so far to all of the affine Cartesian codes.

\begin{theorem}
	Let  $C_X(\rho)$ be an affine Cartesian code with $X=A_1\times\dots\times A_{\mu}$ of degree $\rho$ where $\mu\ge 3$. Let $\nu(\rho):=|\{i\in [\mu]\mid \rho+1<|A_i|\}|$. If $\nu(\rho)\ge 3$, then $C_X(\rho)$ is a batch code capable of satisfying any $t=\nu(\rho)+1$.
\end{theorem}
\begin{proof}

    For simplicity of notation, let $\nu:=\nu(\rho)$. Without loss of generality, under a change of variable,  we can consider $\{i\in [\mu]\mid \rho+1<|A_i|\}=[\nu]$. Consider the puncturing function $\phi:\F_q^\mu\rightarrow \F_q^{\nu}$ that is obtained by puncturing the coordinates in the subset $[\mu]\setminus [\nu]$.
    Let $\calV=\langle(1,\dots,1)\rangle\subseteq \F_q^\mu$ and define the buckets for $\calC_{\phi(X)}(\rho)$ to be the sets $[\tilde{p}]^{\phi(X)} :=[\tilde{p}]\cap \phi(X)$ for $\tilde{p}\in \phi(X)$ and where $[\tilde{p}]\in \F_q^{\nu}/\phi(\calV)$. We will show that for a query $Q=(p_1,\dots,p_{\nu+1})\in X^{\nu+1}\subseteq (\F_q^\mu)^{\nu+1}$ for which we wish to recover $f(p_1),\dots,f(p_{\nu+1})$, there exists a valid recovery set $\{R_{p_1,i_1},\dots R_{p_{\nu+1},i_{\nu+1}}\}$ where $R_{p_{\ell},i_{\ell}}\subseteq X$ with $i_\ell\in I\cup\{0\}$ for all $\ell\in [\nu+1]$.
    
     Consider the affine Cartesian code $\calC_{\phi(X)}(\rho)$. We show that this code has the same batch properties as the code $\calC_{\F_q^{\nu}}(\rho)$. 
    Let \[\phi(Q)=(\phi(p_1), \dots, \phi(p_{\nu+1}))\subseteq (\phi(X))^{\nu+1}\subseteq (\F_q^{\nu})^{\nu+1}\]
     by Theorem \ref{diagBatch} there exists a query recovery set $R^{\phi(Q)}=\{R_{\phi(p_1),i_1},\dots,R_{\phi(p_{\nu+1}),i_{\nu+1}}\}$ for $\phi(Q)$ in $\calC_{\F_q^\nu}$
	with $R_{\phi(p_s),i_s}\subseteq \F_q^{\nu}$ as in Definition \ref{qRecovery} and $i_1,\dots,i_{\nu+1}\in [\nu]\cup\{0\}$.
	
	For all $s\in[\nu+1]$, $R^{\phi(X)}_{\phi(p_s),i_s}=R_{\phi(p_s),i_s}\cap \phi(X)$  matches the definition in Lemma \ref{recoverable}.
	Since $\rho+1<|A_i|$ for all $i\in[\nu]$, by that lemma the values of $f(R^{\phi(X)}_{\phi(p_\ell),i_\ell})$  are enough to recover $f(\phi(p_\ell))$, for any $\ell\in [\nu+1]$ and $f\in \F_q[x_1,\dots,x_{\nu}]^{\le\rho}$. This means that
	\[R^{\phi(Q)}=\{R^{\phi(X)}_{\phi(p_1),i_1},\dots,R^{\phi(X)}_{\phi(p_{\nu+1}),i_{\nu+1}}\}\]
	is a query recovery set for $\phi(Q)$ in $\calC_{\phi(X)}(\rho)$. Furthermore, we have that  by Condition \eqref{intersection} for $\calC_{\F_q^{\nu}}(\rho)$, it holds that
	
	\[\left|\left(\bigcup_{s=1}^{\mu+1}R^{\phi(X)}_{\phi(p_s),i_s}\right)\cap [\tilde{p}]^{\phi(X)}\right|=\left|\left(\bigcup_{s=1}^{\mu+1}R_{\phi(p_s),i_s}\right)\cap [\tilde{p}]\cap \phi(X)\right|\le1\]
	for all $\tilde{p}\in \phi(X)$ and by Condition \eqref{empty} for $\calC_{\F_q^{\nu}}(\rho)$ it holds that 
	
	\[R^{\phi(X)}_{\phi(p_r),i_r}\cap R^{\phi(X)}_{\phi(p_s),i_s}=R_{\phi(p_r),i_r}\cap R_{\phi(p_s),i_s}\cap \phi(X)=\emptyset\]
	for all $r,s\in [\nu+1]$ with $r\neq s$. 
	
	We have thus shown that the query $\phi(Q)=\{\phi(p_1),\dots, \phi(p_{\nu+1})\}\subseteq \phi(X)$ can be recovered in $\calC_{\phi(X)}(\rho)$ using the recovery sets        \begin{equation}\label{equation}
	    \{R^{\phi(X)}_{\phi(p_1),i_1},\dots,R^{\phi(X)}_{\phi(p_{\nu+1}),i_{\nu+1}}\}. 
	\end{equation}
	
	We state that the set $\{R^{X}_{p_1,i_1},\dots,R^{X}_{p_{\nu+1},i_{\nu+1}}\}$
    is a valid recovery set for $Q$ in $\calC_X(\rho)$.  Let $R^Q=\{R_{p_1,i_1},\dots,R_{p_{\nu+1},i_{\nu+1}}\}$ be the set of recovery sets in $\F_q^\mu$ where the indices correspond to the ones of Equation \eqref{equation}. Let $\overline{Q}$ be a query obtained by appending $\mu-\nu$ points of $\F_q^\mu$ to $Q$. Using Theorem  \ref{Expanding}, a query recovery set $R^{\overline{Q}}$ can be constructed recursively by starting from $R^{\phi(Q)}$ such that $R^Q\subseteq R^{\overline{Q}}$ and by taking only the recovery set for the points in $Q$ and restricting them to $X$ we obtain the set $R^Q$.

\end{proof}

\section{Conclusions}

The work in this paper focuses on the study of batch properties of affine Cartesian codes. For affine Cartesian codes, given a subspace $\calV\subseteq \F_q^\mu$, we define a bucket configuration where each bucket is a cosets of $\calV$ in the quotient space $\calV\subset\F_q^{\mu}$. We show that for such bucket configuration to define a batch code, one needs to have that the intersection $\calV\cap\langle e_i,e_j\rangle$ is trivial for all $i\neq j$. By choosing $\calV=\langle(1,\dots,1)\rangle$, we demonstrate that the affine Cartesian code $\calC_{\F_q^{\mu}}{\rho}$ can satisfy queries of length $t=\mu+1$ for any $\mu\ge 3$. We also able to extend the result for any affine Cartesian code, where the size of the query depends on the total degree and the size of the subsets defining the code.

\bibliography{bibliography}
\bibliographystyle{siam}

\end{document}